%% file: Journal_single_column.tex
\documentclass[12pt, draftclsnofoot, onecolumn]{IEEEtran}

\usepackage{tikz}
\usetikzlibrary{calc}
\usepackage{cite}
\usepackage{graphicx}
\usepackage{algorithm}
\usepackage{algorithmic}
\usepackage{epsfig}
\usepackage{epstopdf}
\usepackage{amssymb,amsmath}
\usepackage{amsmath}
\usepackage{amsthm}
\usepackage{amsfonts}
\usepackage{subfigure}
\usepackage{psfrag}
\usepackage{rotating}
\usepackage{latexsym}
\usepackage{dsfont}
\usepackage{stmaryrd}
\usepackage{amsthm}
\usepackage{amssymb}
\usepackage{amsmath}

\usepackage[force,almostfull]{textcomp}
\usepackage{fancyvrb}
\usepackage{textcomp}
\usepackage{url}
\usepackage{ifdraft}
\usepackage{url}
\usepackage{multirow}
\usepackage{rotating}
\usepackage{xspace}
\usepackage{array}
\usepackage{multido}
\usepackage{flushend}
 \usepackage{enumitem}
\usepackage{enumerate}
\usepackage[latin1]{inputenc}
\newlength\figureheight 
\newlength\figurewidth 
\usepackage{graphics}
\usepackage{pgfplots}
\pgfplotsset{compat=newest}
\pgfplotsset{plot coordinates/math parser=false}

\usepackage{scalefnt}

\newtheoremstyle{specialcasestyle}{1mm}{1mm}{\upshape}{}{\bfseries\upshape}{.}{0mm}{}
\theoremstyle{specialcasestyle}

\newtheorem{prop}{Proposition}

\newtheorem{rem}{Remark}

\begin{document}

\title{On the Sum of Order Statistics and Applications to Wireless Communication Systems Performances}

\author{
\IEEEauthorblockN{Nadhir Ben Rached, Zdravko Botev, Abla Kammoun, Mohamed-Slim Alouini, and Raul Tempone}\\

\thanks{
A part of this work has been submitted to IEEE International Conference on Acoustics, Speech and Signal Processing (ICASSP 2018).

Nadhir Ben Rached, Abla Kammoun,  Mohamed-Slim Alouini, and Raul Tempone are with King Abdullah University of Science and Technology (KAUST), Thuwal, Makkah Province, Saudi Arabia, Email: \{nadhir.benrached, abla.kammoun, slim.alouini, raul.tempone\}@kaust.edu.sa

Zdravko Botev is with the University of New South Wales (UNSW), Sydney, NSW, Australia, Email: botev@unsw.edu.au
}
}
\date{}
\maketitle
\thispagestyle{empty}
\begin{abstract}
 We consider the problem of evaluating the cumulative distribution function (CDF) of the sum of order statistics, which  serves to compute outage probability (OP) values at the output of generalized selection combining receivers. Generally, closed-form expressions of the CDF of the sum of order statistics are unavailable for many practical distributions. Moreover, the naive Monte Carlo (MC) method requires a substantial computational effort when the probability of interest is sufficiently small. In the region of small OP values, we propose instead two effective variance reduction techniques that yield a reliable estimate of the CDF with small computing cost. The first estimator, which can be viewed as an importance sampling estimator, has bounded relative error under a certain assumption that is shown to hold for most of the challenging distributions. An improvement of this estimator is then proposed for the Pareto and the Weibull cases. The second  is a conditional MC estimator that achieves  the bounded relative error property for the Generalized Gamma case and the logarithmic efficiency in the Log-normal case. Finally, the efficiency of these estimators is compared via various numerical experiments.
\end{abstract}

\section{Introduction}
Order statistics play an important role in the performance analysis of wireless communication systems over fading channels \cite{Yang:2011:OSW:2132708}. For instance, in the generalized selection combining (GSC) model combined with maximum ratio combining (MRC) diversity technique, the output signal-to-noise-ratio (SNR) is expressed as the partial sum of ordered channel gains, i.e. squares of the amplitudes of the fading channels. More specifically, this scheme selects and combines the $L$ largest SNRs among a total of $N$ diversity branches \cite{895025}. The GSC diversity scheme combined with MRC is then a generalization of MRC and selection combining (SC) diversity techniques. The partial sum of  order statistics is also encountered when GSC is combined with equal gain combining (EGC) diversity technique. In fact, the outage probability (OP) under this model turns out to be equivalent to evaluating the cumulative distribution function (CDF) of the sum of ordered channel amplitudes variates \cite{7328688}. Therefore, from these two examples, it is of major practical interest to evaluate the CDF of the sum of ordered random variables (RVs) as it can serve to compute OP values at the output of GSC diversity receivers combined with either MRC or EGC. 

Closed-form expressions of the CDF of the partial sum of order RVs exist only for particular distributions.  In \cite{5605378}, a unified moment generating function approach has been derived to determine the joint statistics of partial sums of ordered RVs and in particular closed-form expressions have been presented for the exponential RV. A further work on the joint statistics of partial sums of ordered exponential RVs, useful for instance for the analysis of OP of GSC receivers subject to self-interference, has been developed in \cite{5784177}. Based on an equivalent methodology to \cite{5605378}, closed-form results on partial sums of ordered Gamma variates have been developed in \cite{7953495} which in particular applies to OP computation at the output of GSC combined with MRC receivers over the Nakagami fading channel. Further order statistics results in the Nakagami fading model are in \cite{6596692,895025}.

In the particular scheme where all ordered RVs are combined, i.e. this corresponds to the case $L=N$, the CDF of the sum of either channel gains (MRC) or channel amplitudes (EGC) has been extensively studied in the literature. Closed-form expressions of OP at the output of MRC diversity receivers exist for particular fading models such as independent Nakagami-m \cite{6292935} and independent and identically distributed (i.i.d) $\kappa-\mu$ and $\eta-\mu$ \cite{4231253}. Moreover, closed-form approximations have been proposed for sum of Log-normal \cite{1275712,4814351}, Weibull \cite{1665128}, and Rayleigh \cite{1388722} distributions. On the other hand, efficient simulation methods have been also developed for the estimation of the CDF of the sum of RVs such as the Log-normal \cite{7328688,asmussen2016exponential,gulisashvili2016,botev_SLN,Nadhir_SLN} and the Generalied Gamma \cite{7328688}. 

In the general case where $L<N$  and apart from the exponential and Gamma RVs, closed-form expressions of the CDF of partial sums of ordered RVs are out of reach for many challenging distributions and are still open problems. This is for instance the case of the Log-normal RV which models shadowing \cite{Stuber:2001:PMC:368633}  and weak-to-moderate turbulence channels in free space optical communication systems \cite{journals/twc/NavidpourUK07}. The Weibull variate, which has also received an increasing interest and has been shown to fit realistic propagation channels \cite{1512431}, is another example where the CDF of sums of order statistics is not known to possess a closed-form expression. Thus, it is important to propose alternative approaches to compute the CDF of sums of ordered RVs with arbitrary distributions. 

The use of naive Monte Carlo (MC) method can constitute a good alternative to estimate the CDF of partial sums of ordered RVs. However, since for typical wireless communication systems, more attention is accorded to small OP values, i.e. left-tail of the CDF of the sum of ordered RVs, naive MC method is known to be computationally expensive, requiring a substantial amount of samples to yield an accurate estimate of the left-tail of the CDF. This motivates our work in which we aim to propose efficient variance reduction MC techniques that yield  very precise estimate of the CDF of the sum of ordered RVs with small computing cost \cite{opac-b1132466}. The main contributions of our paper are summarized as follows:
\begin{itemize}
\item We provide a universal importance sampling (IS) estimator \cite{opac-b1132466}  and show that it has bounded relative error, a relevant property in the context of rare event simulation, under a mild assumption that is shown to hold for many challenging distributions. A non-exhaustive list includes for instance the Generalized Gamma (and in particular the Gamma and the Weibull distributions), and the $\kappa-\mu$ distributions (which includes the Rice distribution as a particular case). An improvement on the universal IS estimator is proposed for two particular scenarios: the Pareto and the Weibull (with shape parameters between zero and one).   
\item We propose a second estimator based on the use of conditional MC approach and show that it achieves the bounded relative error property for the Generalized Gamma case and the logarithmic efficiency, a weaker property than the bounded relative error, for the Log-normal case. 
\item We identify the regions in which the IS estimators outperform, in term of computational effort measured by the variance of each estimators, the conditional MC estimator and vice versa. Moreover, the smoothness of the conditional MC estimator enables us the further improve its convergence rate  via the use of the Quasi MC method. 
\end{itemize}
The rest of the paper is organized as follows. In Section II, we describe the problem setting and define the main concepts. The universal IS estimator is presented in Section III. In the same section, we present an improved variant of this estimator for the Pareto and the Weibull scenarios. In Section IV, an alternative estimator based on the use of conditional MC is described.  Finally, some selected numerical results are shown in Section V to compare the performances of the proposed estimators. 
\section{Problem Setting}
We consider a sequence of i.i.d RVs $X_1,X_2,\cdots,X_N$ with common probability density function (PDF) $f(\cdot)$. Our objective is to propose efficient MC methods to evaluate the following quantity
\begin{align}\label{op}
\ell=P \left ( \sum_{k=1}^{L}{X^{(k)}} \leq \gamma_{th}\right ),
\end{align}
where $\gamma_{th}$ is the threshold value, $X^{(k)}$ represents the $k^{th}$ order statistic such that $X^{(1)} \geq X^{(2)} \geq \cdots \geq X^{(N)}$, and $L$ is an integer satisfying $1 \leq L \leq N$. 
The above expression of $\ell$ is a useful metric in the performance analysis of wireless communication systems, operating over fading channels. An example of application is that of transmissions between a single-antenna transmitter and an $N$-antennas receiver. Then, the quantity $\sum_{k=1}^{L}{X^{(k)}}$ corresponds to the total SNR when the receiver selects the $L$ best individual SNR reaching each of the diversity branches. In this case, the quantity $\ell$ corresponds to the OP at the output of GSC combined with MRC receivers. In particular, when $L=1$, the expression in (\ref{op}) corresponds to the OP at the output of SC receivers and to the OP at the output of MRC diversity receivers when $L=N$.

Unfortunately, a closed-form expression of $\ell$ is generally out of reach for many challenging distributions including, for instance,  the Log-normal and the Generalized Gamma. An alternative approach to approximate $\ell$ is then through the use of naive MC simulations. However, it is well-known that for small values of $\ell$, which is the case in typical wireless communication systems, the naive MC method is not practical, since it requires a substantial number of simulations to ensure a precise estimate of $\ell$. Variance reduction techniques can deliver a reliable estimate of $\ell$ with fewer number of runs compared to naive MC simulations. Before delving into the core of our paper, it is important to define some performance metrics that serve to measure the efficiency of an unbiased estimator \cite{opac-b1132466,opac-b1123521}. Let $\hat {\ell}$ be an estimator of $\ell$ with $\mathbb{E} [ \hat{\ell}]=\ell$, we say that $\hat{\ell}$ is logarithmic efficient  when
\begin{align}\label{log_eff}
\lim_{\gamma_{th} \rightarrow 0} \frac{\log \left ( \mathbb{E} \left [ \hat{\ell}^2\right ]\right )}{\log \left ( \ell\right )}=2,
\end{align}
or equivalently for all $\epsilon>0$
\begin{align}
\lim_{\gamma_{th} \rightarrow 0}\frac{\mathrm{var} \left [ \hat{\ell}\right ]}{\ell^{2-\epsilon}}=0.
\end{align}
Note that the limit in (\ref{log_eff}) cannot be made larger since $\frac{\log \left ( \mathbb{E} \left [ \hat{\ell}^2\right ]\right )}{\log \left ( \ell\right )}$ is always less than $2$ from Jensen's inequality. A stronger criterion than the logarithmic efficiency is the bounded relative error which holds when
\begin{align}
\limsup_{\gamma_{th} \rightarrow 0}\frac{\mathrm{var} \left [\hat{\ell} \right ]}{\ell^2} < \infty.
\end{align}
Such a property implies that the number of samples needed to achieve a given accuracy remains bounded regardless of how small $\ell$ is. Finally, a further stronger criterion is the asymptotically vanishing relative error property:
\begin{align}
\limsup_{\gamma_{th} \rightarrow 0}\frac{\mathrm{var} \left [\hat{\ell} \right ]}{\ell^2}=0.
\end{align}
When this criterion holds, the number of simulation runs to meet an accuracy requirement gets smaller as $\ell$ decreases.  
\section{Importance Sampling Estimator}
In this section we present our first estimator of $\ell$. Let $\bold{X}=(X_1,\cdots,X_N)'$ and $S=\{\boldsymbol{x}=(x_1,\cdots,x_N)': \sum_{k=1}^{L}{x^{(k)}} \leq \gamma_{th} \}$ and consider another set $S_1$ that includes $S$ with the assumption that $P \left (\bold{X} \in S_1 \right )$ is known in closed form. Then, the probability $\ell$ is re-written as
\begin{align}
\ell= P \left (\bold{X} \in S \right )=P \left (\bold{X} \in S_1 \right ) P \left ( \bold{X} \in S  | \bold{X} \in S_1\right ).   
\end{align}
Hence, an estimator of $\ell$ is given by the use of naive MC simulation to estimate $P \left ( \bold{X} \in S  | \bold{X} \in S_1\right )$. More specifically, from the above expression, we may write $\ell$ as
\begin{align}
\ell= \mathbb{E}_g \left [\ell_1\bold{1}_{\left ( \bold{X} \in S\right )} \right ]\triangleq \mathbb{E}_g \left [ \hat{\ell}_{IS}\right ],
\end{align}
where $g(\cdot)$ is the PDF under which $\bold{X}$ is distributed according to its original PDF truncated over $S_1$, $\ell_1$ is equal to $P \left (\bold{X} \in S_1 \right )$, and $1_{\left (\cdot \right )}$ is the indicator function. It is worth mentioning that $\hat{\ell}_{IS}$ may be viewed as an IS estimator with biasing PDF $g(\cdot)$. 

Now, we discuss how to select $S_1$ in order to achieve a substantial amount of variance reduction. Intuitively, the set $S_1$ has to be selected such that $\ell_1$ is close to $\ell$ since the variance of $\hat{\ell}_{IS}$ is given by
\begin{align}\label{var_is}
\mathrm{var}_g \left [\hat{\ell}_{IS} \right ]=\ell_1 \ell-\ell^2.
\end{align}
Thus, we clearly point out that the closer $\ell_1$ to $\ell$, the smaller the variance of $\hat{\ell}_{IS}$ is, and hence the more efficient is the estimator $\hat{\ell}_{IS}$. In particular, the estimator $\hat {\ell}_{IS}$ has  bounded relative error  when $\ell_1/\ell$ is asymptotically bounded as $\gamma_{th}$ goes to $0$, and has asymptotically vanishing relative error in the case where $\ell_1/\ell$ approaches $1$ as $\gamma_{th}$ goes to $0$. 

In the next subsection, we propose the simplest choice of $S_1$ that has the feature of being applicable to any distribution and prove that the bounded relative error holds under a mild assumption that is valid for most of the challenging distributions. 
\subsection{Universal IS Estimator}
The simplest choice of the set $S_1$ is as follows
\begin{align}
S_1=\{\boldsymbol{x}=(x_1,\cdots,x_N)': x^{(1)} \leq \gamma_{th} \}.
\end{align}
The probability $\ell_1$ is therefore given by
\begin{align}
\ell_1= \left (P \left (X_1 \leq \gamma_{th} \right ) \right )^N
\end{align}
The efficiency of this IS estimator is given in the following proposition
\begin{prop}
\hspace{2mm} For distributions satisfying $P \left (X_1<\gamma_{th} \right )/P \left ( X_1 \leq \gamma_{th}/L\right )=\mathcal{O}(1)$ as $\gamma_{th} \rightarrow 0$, we have
\begin{align}
\limsup_{\gamma_{th} \rightarrow 0}\frac{\ell_1}{\ell} < \infty
\end{align}
Hence, the bounded relative error property holds.
\end{prop}
\begin{proof}
Let us first lower bound the probability of interest $\ell$ as follows
\begin{align}
\ell \geq P \left (X_1 \leq \gamma_{th}/L, \cdots, X_N \leq \gamma_{th}/L \right )
\end{align}
Hence, we get
\begin{align}
\frac{\ell_1}{\ell} \leq \frac{\left (P \left (X_1 \leq \gamma_{th} \right ) \right )^N}{P \left (X_1 \leq \gamma_{th}/L, \cdots, X_N \leq \gamma_{th}/L \right )}
\end{align}
\end{proof}
The assumption $P \left (X_1<\gamma_{th} \right )/P \left ( X_1 \leq \gamma_{th}/L\right )=\mathcal{O}(1)$ is not restrictive since it is satisfied by many challenging distributions such that the Generalized Gamma (which includes in particular the Gamma and the Weibull distributions), and the $\kappa-\mu$ distributions, see \cite{7769235}. Moreover, in the independent and not identically distributed scenario, the bounded relative error property holds when the assumption of Proposition 1 is satisfied for each $X_i$, $i=1,\cdots,N$. In particular when $L=N$, this IS estimator, with the assumption in Propostion 1, is the first to achieve the bounded relative error property in the independent and not identically distributed case since, to the best of the authors' knowledge, this property has only been achieved  in the i.i.d setting \cite{7328688}. 

Despite its general scope of applicability, the efficiency of this universal IS estimator can be improved if we settle for a particular distribution.  This is the aim of the two following subsections where we propose other choices of $S_1$ in the Pareto and Weibull cases that improve the efficiency of the universal IS estimator. 
\subsection{Pareto Case}
\subsubsection{The Approach}
The PDF $f(\cdot)$ of $X_i$, $i=1,\cdots,N$, is given as
\begin{align}
f(x)=\alpha \left (1+x \right )^{-(1+\alpha)}, \text{    } x\geq 0,
\end{align}
with $\alpha>0$. It is easy to observe that if we define $Y_i=\alpha \log \left ( 1+X_i\right )$, $i=1,\cdots,N$, then $Y_i$ has an exponential distribution with mean $1$. Using this transformation, $\ell$ is re-written as follows
\begin{align}
\ell=P \left (\sum_{k=1}^{L}{\exp \left (Y^{(k)}/\alpha \right )} \leq \gamma_{th}+L \right ).
\end{align}
Now, we will take advantage of the convexity of the exponential function to construct the set $S_1$. Let $\lambda_i>0$ such that $\sum_{i=1}^{L}{\lambda_i}=1$, then we get
\begin{align}
\nonumber &\sum_{k=1}^{L}{\lambda_k \exp \left (Y^{(k)}/\alpha-\log \left (\lambda_k \right ) \right )} \\
&\geq \exp \left ( \sum_{k=1}^{L}{\lambda_k \left ( Y^{(k)}/\alpha-\log \left ( \lambda_k\right )\right )}\right ).
\end{align}
Hence, the set $S_1$ is selected as 
\begin{align}
\nonumber S_1&= \Big \{\boldsymbol{y}=(y_1,\cdots,y_N)':\sum_{k=1}^{L}{\lambda_k y^{(k)}} \\
&\leq \alpha ( \log(\gamma_{th}+L) +\sum_{k=1}^{L}{\lambda_k \log (\lambda_k)})\Big \}.
\end{align}
The remaining work is to compute $\ell_1$ and to provide a procedure on how to generate samples according to $g(\cdot)$. By denoting  $\gamma_1=\alpha \left (\log(\gamma_{th}+L) +\sum_{k=1}^{L}{\lambda_k \log (\lambda_k)}\right )$ and exploiting the following representation of the order statistics $Y^{(1)},\cdots, Y^{(L)}$, see \cite{opac-b1132466}
\begin{align}\label{transform}
Y^{(k)}=\sum_{j=1}^{N-k+1}{\frac{Z_j}{N-j+1}},
\end{align}
where $Z_1,\cdots,Z_N$ are i.i.d exponential RVs with mean $1$, it follows that $\ell_1$ is given by
\begin{align}
\ell_1=P \left ( \sum_{i=1}^{N}{\beta_i Z_i} \leq \gamma_1\right ),
\end{align}
where 
\begin{align}
\beta_i = \begin{cases}
\sum\limits_{j=1}^{L}{\lambda_j}/(N-i+1) &\text{   }i=1=1,\cdots,N-L+1,\\
\sum\limits_{j=1}^{N+1-i}{\lambda_j}/(N-i+1) &\text{    }i=N-L+2,\cdots,N.
\end{cases}
\end{align}
Hence, $\ell_1$ turns out to be the CDF of the sum of independent exponential RVs. A closed-form expression of $\ell_1$ is as follows, see \cite{Botev:2013:SNR:2466677.2466683},
\begin{align}
\ell_1=1-(1,0,\cdots,0)\exp \left ( \gamma_1 \bold{A} \right ) (1,1,\cdots,1)',
\end{align}
with $\exp \left ( \gamma_1 \bold{A} \right )$ being the matrix exponential of $\gamma_1 \bold{A}$ and 
\begin{align}
\bold{A} = 
 \begin{pmatrix}
  -1/\beta_1 & 1/\beta_1     & 0    & \cdots & 0 \\
    0       & -1/\beta_2   & 1/\beta_2 & \cdots & 0 \\
  \vdots  & \vdots  & \ddots & \ddots &\vdots  \\
    0 & \cdots & 0 & -1/\beta_{N-1} &1/\beta_{N-1} \\
  0 & \cdots & 0 & 0 &-1/\beta_N 
 \end{pmatrix}
\end{align}
Now, we answer the question on how we generate samples truncated over the set $S_1$. To this end, we use the representation (\ref{transform}) and sample  $Z_1, \cdots, Z_N$, which are exponentially distributed with mean $1$, conditional on the event $\{\sum_{i=1}^{N}{\beta_i Z_i} \leq \gamma_1\}$. This can be efficiently performed by letting $T_i=\beta_i Z_i/\gamma_1$ and using a uniform distribution over $\{\sum_{i=1}^{N}{T_i} \leq 1\}$ as acceptance-rejection proposal. The following algorithm provides all steps to sample $Z_1, \cdots, Z_N$ conditional on the event  $\{\sum_{i=1}^{N}{\beta_i Z_i} \leq \gamma_1\} $, or equivalently to sample  $Y^{(1)}, \cdots, Y^{(L)}$, restricted to $S_1$.

\begin{algorithm}[H]
\caption{Samples Truncated over $S_1$}
\begin{algorithmic}[1]\label{Algo1}
\STATE \textbf{Inputs:} $\{\beta_i \}_{i=1}^{N}$, $\gamma_1$.
\STATE \textbf{Outputs:} $\{Y^{(i)}\}_{i=1}^{L}$.
\WHILE {$U > \exp \left (-\gamma_1\sum_{i=1}^{N}{U_i/\beta_i} \right )$} 
\STATE Generate $\{U_i\}_{i=1}^{N}$  from the uniform distribution over the set $\{u_i \geq 0, \sum_{i=1}^{N}{u_i}\leq 1\}$, see \cite[Algorithm 3.23]{opac-b1132466}.
\STATE Generate $U$ a sample from the uniform distribution over $[0,1]$.
\ENDWHILE
\STATE $\bold{T} \leftarrow \bold{U}$
\STATE Set $Z_i \leftarrow(\gamma_1/\beta_i) T_i$.
\STATE Compute $\{Y^{(k)} \}_{k=1}^{L}$ from (\ref{transform}).
\end{algorithmic}
\end{algorithm}

\subsubsection{Efficiency}
We investigate in this part the efficiency of the proposed IS scheme. The main result is in the following proposition.
\begin{prop}
\hspace{2mm} Let $\lambda_k=1/L$ for all $k \in \{1,\cdots,L \}$. Then, we have
\begin{align}
\limsup_{\gamma_{th} \rightarrow 0}{\frac{\ell_1}{\ell}} < \infty.
\end{align}
Thus, the bounded relative error property holds.
\end{prop}
\begin{proof}
Let us upper bound $\ell_1$ as follows
\begin{align}
\nonumber \ell_1&=P \left ( \sum_{k=1}^{L}{Y^{(k)}} \leq \alpha L \left [ \log \left (1+\gamma_{th}/L \right )\right ]\right )\\
\nonumber & \leq P \left (Y^{(1)} \leq  \alpha L \left [ \log \left (1+\gamma_{th}/L \right )\right ]\right )\\
&= \left (1-\exp \left ( -\alpha L \log \left ( 1+\gamma_{th}/L\right )\right ) \right )^{N}.
\end{align}
Now, the probability $\ell$ is lower bounded as follows
\begin{align}
\nonumber \ell &=P \left ( \sum_{k=1}^{L}{X^{(k)}} \leq \gamma_{th}\right )\\
\nonumber  &\geq P \left ( X^{(1)} \leq \gamma_{th}/L,\cdots, X^{(L)} \leq \gamma_{th}/L\right )\\
\nonumber &= P \left (X_1 \leq \gamma_{th}/L,\cdots, X_N \leq \gamma_{th}/L \right )\\
&= \left (1-\exp \left (-\alpha \log \left ( 1+\gamma_{th}/L \right )\right ) \right )^{N}.
\end{align}
Therefore, we deduce that 
\begin{align}
\limsup_{\gamma_{th} \rightarrow 0 }{\frac{\ell_1}{\ell}} \leq L^N,
\end{align}
and hence the proof is concluded.
\end{proof}
\subsection{Weibull Case}
\subsubsection{The Approach}
we consider the case where $X_1,\cdots,X_N$ are i.i.d Weibull variates with PDF
\begin{align}
f(x)=\frac{\alpha}{\eta} \left ( \frac{x}{\eta}\right )^{\alpha-1} \exp \left (-\left ( \frac{x}{\eta}\right )^{\alpha} \right ),\text{   } x>0,
\end{align}
where $\eta$ is the scale parameter, $\alpha$ is the shape parameter which is assumed, in this part, to satisfy $0<\alpha<1$. Consider now the RVs $Y_i=\left (X _i/\eta\right )^{\alpha}$, $i=1,\cdots,N$. Then, it easy to show that $Y_i$, $i=1,\cdots,N$ are i.i.d exponential RVs with mean $1$. Hence, $\ell$ is re-expressed as
\begin{align}\label{weib}
\ell=P \left (\sum_{k=1}^{L}{\left (Y^{(k)} \right )^{1/\alpha}} \leq \gamma_{th}/\eta \right ).
\end{align}
Let $\lambda_i>0$, $i=1,\cdots,L$, such that $\sum_{i=1}^{L}{\lambda_i}=1$. Then, using the convexity of $y \rightarrow y^{1/\alpha}$ on the positive axis for $0<\alpha<1$, we get
\begin{align}
\nonumber & \left \{\sum_{k=1}^{L}{\lambda_k\left (Y^{(k)}/\lambda_k^{\alpha} \right )^{1/\alpha}} \leq \gamma_{th}/\eta\right \} \\
&\subseteq \left \{ \left (\sum_{k=1}^{L}{\lambda_k^{1-\alpha}Y^{(k)}} \right )^{1/\alpha} \leq \gamma_{th}/\eta \right \}.
\end{align}
Therefore, $S_1$ is selected as 
\begin{align}
S_1=\left \{ \bold{y}=(y_1,\cdots,y_N)':\sum_{k=1}^{L}{\lambda_k^{1-\alpha}Y^{(k)}} \leq \left ( \gamma_{th}/\eta\right )^{\alpha}\right \}.
\end{align} 
Using the same idea as in the Pareto case, the value of $\ell_1$ is written as
\begin{align}
\ell_1=P \left (\sum_{i=1}^{N}{\nu_iZ_i} \leq \left ( \gamma_{th}/\eta\right )^{\alpha} \right ),
\end{align}
with \begin{align}
\nu_i = \begin{cases}
\sum\limits_{j=1}^{L}{\lambda_j^{1-\alpha}}/(N-i+1) &\text{   }i=1=1,\cdots,N-L+1,\\
\sum\limits_{j=1}^{N+1-i}{\lambda_j^{1-\alpha}}/(N-i+1) &\text{    }i=N-L+2,\cdots,N.
\end{cases}
\end{align}
Thus, a closed-form formula for $\ell_1$ is given as
\begin{align}
\ell_1=1-(1,0,\cdots,0)\exp \left ( \gamma_2 \bold{A} \right ) (1,1,\cdots,1)',
\end{align}
with $\gamma_2=\left ( \gamma_{th}/\eta\right )^{\alpha}$. 
Finally, to sample $Y^{(1)},\cdots,Y^{(L)}$ from the truncated PDF over $S_1$, acceptance-rejection is again used and yields an algorithm similar to Algorithm 1.
\subsubsection{Efficiency}
The main result is provided as follows:
\begin{prop}
\hspace{2mm}  For $0<\alpha<1 $ and arbitrary values of $\lambda_k$, $k=1,\cdots,L$, we have
\begin{align}
\limsup_{\gamma_{th} \rightarrow 0}{\frac{\ell_1}{\ell}} < \infty.
\end{align}
Hence, the bounded relative error property holds.
\end{prop}
\begin{proof}
We use the same steps as in the proof of bounded relative error for Pareto case. In fact, the value of $\ell_1$ satisfies
\begin{align}
\nonumber \ell_1 &\leq P \left ( Y^{(1)} \leq \left ( \gamma_{th}/\eta\right )^{\alpha} /\lambda_1^{1-\alpha} \right )\\
&= \left ( 1-\exp \left (- \left ( \gamma_{th}/\eta\right )^{\alpha} /\lambda_1^{1-\alpha} \right )\right )^{N}.
\end{align}
On the other hand, we have
\begin{align*}
\nonumber \ell &\geq P \left ( X^{(1)} \leq \gamma_{th}/L, \cdots, X^{(L)} \leq \gamma_{th}/L\right )\\
\nonumber &=P \left ( X_1 \leq \gamma_{th}/L, \cdots, X_N \leq \gamma_{th}/L\right )\\
&= \left ( 1-\exp \left (- \left ( \gamma_{th}/(L\eta)\right )^{\alpha}  \right )\right )^{N}.
\end{align*}
Thus, we get
\begin{align}
\limsup_{\gamma_{th} \rightarrow 0}{\frac{\ell_1}{\ell}} \leq \frac{L^{\alpha N}}{\lambda_1^{N(1-\alpha)}}.
\end{align}
\end{proof}
Note that, in contrast to the Pareto case where the bounded relative error property holds only for equal values of $\lambda_k$, $k=1,\cdots,L$, the bounded relative error holds in the Weibull case for arbitrarily values of $\lambda_k$ satisfying $\lambda_k>0$ and $\sum_{k=1}^{L}{\lambda_k}=1$. Thus, the values of $\lambda_k$ can be optimized in order to achieve the largest amount of variance reduction. In other words, we may select the values of $\lambda_k$ that minimize the value $\ell_1$ and hence minimize the variance of the estimator $\hat{\ell}_{IS}$. 

\section{Conditional MC Estimator}
The Log-normal distribution is an example for which the assumption in Proposition 1, required to ensure the bounded relative error, is not satisfied. However, we may easily prove in this case that the logarithmic efficiency is achieved by the universal IS estimator. Therefore, it would be important to construct a competitor estimator for the Log-normal case and investigate its efficiency with respect to the universal IS estimator. This is the objective of this section where we propose an alternative estimator of $\ell$ based on the use of conditional MC. In addition to the Log-normal distribution, this conditional MC estimator applies to the Generalized Gamma distribution. For each case, we present the conditional MC estimator along with its corresponding efficiency results.
\subsection{Generalized Gamma Case}
\subsubsection{The Approach}
We start by considering the particular Weibull case. This will facilitate the understanding of the approach in the Generalized Gamma case. From the expression of $\ell$ in (\ref{weib}), the idea of the conditional MC estimator is to use the fact that exponential RV $Y_i$ is equal in distribution to $G\times S_i$ where $G$ is a Gamma distribution with shape $N$ and scale $1$ and $\bold{S}=(S_1,\cdots,S_N)$ are uniformly distributed over the simplex $\{ s_i>0, \sum_{i=1}^{N}{s_i}=1\}$ and independent of $G$, see \cite{opac-b1132466}. Then, using this representation, the probability $\ell$ can be expressed as
\begin{align}
\ell = P \left ( G^{1/\alpha} \left [ \sum_{k=1}^{L}{ \left (S^{(k)} \right )^{1/\alpha}}\right ] \leq \gamma_{th}/\eta\right ).
\end{align}
Let $F_G(\cdot)$ be the CDF of the Gamma RV $G$. By conditioning on $S_1,S_2,\cdots, S_N$, we get
\begin{align}
\ell=\mathbb{E} \left [F_G \left (\frac{(\gamma_{th}/\eta)^\alpha}{\left [\sum_{k=1}^{L}{ \left ( S^{(k)}\right )^{1/\alpha}} \right ]^{\alpha}} \right ) \right ].
\end{align}
Therefore the conditional MC estimator is given by 
\begin{align}\label{cmc_weibull}
\hat{\ell}_{CMC}=F_G \left (\frac{(\gamma_{th}/\eta)^\alpha}{\left [\sum_{k=1}^{L}{ \left ( S^{(k)}\right )^{1/\alpha}} \right ]^{\alpha}} \right ).
\end{align}
The case of the Generalized Gamma distribution is essentially based on the same methodology as above. In fact, let $X_1,X_2,\cdots,X_N$ be a sequence of i.i.d  generalized Gamma RVs whose common PDF is given by
\begin{align}
f(x)=\frac{p/a^d x^{d-1} \exp \left (-\left (x/a \right )^p \right )}{\Gamma(d/p)}, \text{   }x>0.
\end{align} 
It can be easily shown that $Y^{1/p}$, where $Y$ is a Gamma distribution with scale parameter $a^p$ and shape parameter $d/p$, has the same Generalized Gamma distribution. Therefore, the probability $\ell$ is given by 
\begin{align}
\ell =P \left (\sum_{k=1}^{L}{ \left (Y^{(k)} \right )^{1/p}\leq \gamma_{th}} \right ).
\end{align}
Similarly to the Weibull case, we exploit the following representation of the Gamma RVs $Y_1,Y_2,\cdots,Y_N$, see\cite{opac-b1132466}
\begin{align}
Y_i=S_i V, \text{  } i=1,\cdots,N,
\end{align} 
where $\bold{S}=(S_1,\cdots,S_N)$ follows a Dirichlet distribution with parameters $(d/p,\cdots,d/p)$ and $V$ follows a gamma distribution with scale parameter $a^p$ and shape parameter $N d/p$. Note that $S$ and $V$ are independent. Hence, following this representation, the probability of interest can be expressed as 
\begin{align}
\ell= P \left (V^{1/p} \sum_{k=1}^{L}{ \left ( S^{(k)} \right )^{1/p}} \leq \gamma_{th}\right ).
\end{align}
Therefore, by conditioning on $\bold{S}$, it follows that
\begin{align}
\ell= \mathbb{E} \left [F_V \left (\frac{\gamma_{th}^{p}}{ \left ( \sum_{k=1}^{L}{ \left ( S^{(k)}\right )^{1/p}}\right )^p} \right ) \right ],
\end{align}
where $F_V (\cdot)$ is the CDF of the Gamma RV $V$ which is given by
\begin{align}
F_V(x)=\frac{\gamma(Nd/p,x/a^p)}{\Gamma(Nd/p)},
\end{align}
where $\Gamma(\cdot)$ and $\gamma(\cdot,\cdot)$ are respectively the Gamma and the lower incomplete Gamma functions \cite{gradshteyn2007}. Thus, the conditional MC estimator is
\begin{align}
\hat{\ell}_{CMC}=F_V \left (\frac{\gamma_{th}^{p}}{ \left ( \sum_{k=1}^{L}{ \left ( S^{(k)}\right )^{1/p}}\right )^p} \right ).
\end{align}
\subsubsection{Efficiency}
The efficiency of the conditional MC estimator is given in the following proposition. 
\begin{prop}
\hspace{2mm} The conditional MC estimator has bounded relative error for all $1<L \leq N$
\begin{align}
\limsup_{\gamma_{th} \rightarrow 0}{\frac{\mathbb{E} \left [F_V^2 \left (\frac{\gamma_{th}^{p}}{ \left ( \sum_{k=1}^{L}{ \left ( S^{(k)}\right )^{1/p}}\right )^p} \right ) \right ]}{\ell^2}} <\infty.
\end{align}
\end{prop}
\begin{proof}
In a first stage, we start by proving the result for $N=L$. Then, the extension to the general case will be straightforward. Let us first consider the case where $p \geq1$, the second moment of the conditional MC estimator is bounded by
\begin{align}
\nonumber &\mathbb{E} \left [ F_V^2 \left (\frac{\gamma_{th}^{p}}{ \left ( \sum_{k=1}^{N}{ \left ( S_k\right )^{1/p}}\right )^p} \right ) \right ] \\
&\leq F^2_V(\gamma_{th}^p)=\gamma^2(Nd/p,\gamma_{th}^p/a^p)/\Gamma^2(Nd/p),
\end{align}
where we have used the fact that $\sum_{k=1}^{N}{S_k}=1$. Via the use of the asymptotic behavior of the incomplete Gamma function, $\gamma(c,t) \sim c_1 t^c $ as $t \rightarrow 0$ \cite{gradshteyn2007}, we have, for a sufficiently small $\gamma_{th}$,
\begin{align}
\mathbb{E} \left [ F_V^2 \left (\frac{\gamma_{th}^{p}}{ \left ( \sum_{k=1}^{N}{ \left ( S_k\right )^{1/p}}\right )^p} \right ) \right ] \leq C_1 \gamma_{th}^{2Nd},
\end{align}
where $C_1$ is a constant independent of $\gamma_{th}$. On the other hand, the probability $\ell$ is lower bounded as follows
\begin{align}
\nonumber \ell & \geq P \left (X_1 \leq \gamma_{th}/N, \cdots, X_N \leq \gamma_{th}/N \right )\\
& =\left (\gamma(d/p,(\gamma_{th}/Na)^p)/\Gamma(d/p) \right )^N. 
\end{align}
Again, using the asymptotic behavior of the incomplete gamma function, we get for a sufficiently small values of $\gamma_{th}$
\begin{align}\label{ell_eq}
\ell \geq C_2 \gamma_{th}^{Nd},
\end{align}
where $C_2$ is a constant independent of $\gamma_{th}$. Hence, the bounded relative error property holds for $p \geq 1$. In the case where $p<1$, we use the convexity of the function $x \mapsto x^{1/p}$ to get
\begin{align}
\mathbb{E} \left [ F_V^2 \left (\frac{\gamma_{th}^{p}}{ \left ( \sum_{k=1}^{N}{ \left ( S_k\right )^{1/p}}\right )^p} \right ) \right ] \leq F_V^2 \left ( \gamma_{th}^p/N^{p-1}\right ).
\end{align}
Then, we use the same steps as in the case where $p \geq 1$ to conclude the proof. The proof of the bounded relative error for the case $1<L<N$ is straightforward from the above proof. Let us start with the case where $p \geq 1$. The second moment is bounded by 
\begin{align}
\nonumber & \mathbb{E} \left [ F_V^2 \left (\frac{\gamma_{th}^{p}}{ \left ( \sum_{k=1}^{L}{ \left ( S^{(k)}\right )^{1/p}}\right )^p} \right ) \right ] \\
&\leq \mathbb{E} \left [ F_V^2 \left (\frac{\gamma_{th}^{p}}{ \left ( \sum_{k=1}^{L}{  S^{(k)}}\right )^p} \right ) \right ].
\end{align}
Then, using the fact that $\frac{N}{L}\sum_{k=1}^{L}{S^{(k)}} \geq \sum_{k=1}^{N}{S_k}$. It follows that
\begin{align}
\mathbb{E} \left [ F_V^2 \left (\frac{\gamma_{th}^{p}}{ \left ( \sum_{k=1}^{L}{ \left ( S^{(k)}\right )^{1/p}}\right )^p} \right ) \right ] \leq F_V^2 \left ( N^p\gamma_{th}^p/L^p\right ).
\end{align}
The lower bound on $\ell$ in (\ref{ell_eq}) remains valid when $L<N$. Therefore the bounded relative error holds. In the case where $p<1$, we have, using again $\frac{N}{L}\sum_{k=1}^{L}{S^{(k)}} \geq \sum_{k=1}^{N}{S_k}$ and the convexity of the function $x \mapsto x^{1/p}$,
\begin{align}
\mathbb{E} \left [ F_V^2 \left (\frac{\gamma_{th}^{p}}{ \left ( \sum_{k=1}^{L}{ \left ( S^{(k)}\right )^{1/p}}\right )^p} \right ) \right ] \leq  F_V^2 \left ( N\gamma_{th}^p/L^p\right ),
\end{align}
and therefore the bounded relative error property holds again.
\end{proof}
\begin{rem}
\hspace{2mm} In the particular Weibull setting where $p=d=\alpha$, the conditional MC estimator does not impose a restriction on the value of $\alpha$ which can take any strictly positive value. This is in contrast with the IS estimator described in the previous section which assumes $\alpha$ to be between $0$ and $1$. 

It is also important to note that the conditional MC estimator has bounded relative error for $1<L \leq N$ when $p_i=p$, $a_i=a$, whereas $d_i$ is allowed to take arbitrary values. 
\end{rem}
\subsection{Log-normal Case}
\subsubsection{The Approach}
We consider a sequence $X_1,\cdots, X_N$ of i.i.d standard Log-normal RVs whose PDF is
\begin{align}
f(x)=\frac{1}{\sqrt{2\pi}x}\exp \left (- \left ( \log(x)\right )^2/2 \right ), \text{   }x>0.
\end{align}
Let $Y_1,\cdots,Y_N$ be the associated normal RVs with zero mean and unit variance. Then, the probability $\ell$ is expressed as
\begin{align}
\ell=P \left (\sum_{k=1}^{L}{\exp \left ( Y^{(k)} \right ) }\leq \gamma_{th}\right ).
\end{align}
The random vector $\bold{Y}=(Y_1,\cdots,Y_N)'$ can be decomposed as, see \cite{DBLP:conf/wsc/BlanchetJR08},
\begin{align}
\bold{Y}=R \bold{\Theta},
\end{align}
where $R$ is the Euclidean distance of $Y$ from the origin and $\bold{\Theta}$ is uniformly distributed over the surface of the N-dimensional Ball. Note that $R$ and $\bold{\Theta}$ are independent. Following this representation, the probability of interest is expressed as 
\begin{align}
\ell=P \left ( \sum_{k=1}^{L}{\exp \left (R \Theta^{(k)} \right )} \leq \gamma_{th}\right ).
\end{align}
We assume now that $\gamma_{th} \leq 1$. Note that as long as $\ell$ is not sufficiently small, when $\gamma_{th}>1$, the previous assumption is not restrictive since it can be efficiently handled using naive MC simulations. Under this assumption, we clearly observe that for a given realization in which one of the $\Theta_i$ is greater than $0$ then this realization will certainly not be in the set of interest $\{\sum_{k=1}^{L}{\exp \left (R \Theta^{(k)} \right )} \leq \gamma_{th} \}$. Hence, it more convenient to condition on the event $\{\max_{i}{\Theta_i} <0 \}$. More clearly, the probability $\ell$ is written as
\begin{align}
\nonumber \ell &= P \left ( \sum_{k=1}^{L}{\exp \left (R \Theta^{(k)} \right )} \leq \gamma_{th},\max_{i}{\Theta_i} <0  \right )\\
&+ \underbrace{P \left ( \sum_{k=1}^{L}{\exp \left (R \Theta^{(k)} \right )} \leq \gamma_{th},\max_{i}{\Theta_i} \geq 0  \right )}_{=0}.
\end{align}
Given that $P \left ( \max_{i}{\Theta_i} <0 \right )=1/2^N$ and by conditioning over $\bold{\Theta}|\max_i(\Theta_i) <0$, we get
\begin{align}
\nonumber \ell &=\frac{1}{2^N} P \left ( \sum_{k=1}^{L}{\exp \left (R \Theta^{(k)} \right )} \leq \gamma_{th}\Big{|}\max_{i}{\Theta_i} <0  \right )\\
&= \frac{1}{2^N} \mathbb{E}_{\bold{\Theta} \Big {|} \max_i(\Theta_i) <0} \left [ P \left ( \sum_{k=1}^{L}{\exp \left (R \Theta^{(k)} \right )} \leq \gamma_{th}\Big{|} \bold{\Theta}\right )\right ].
\end{align}
Given $\bold{\Theta}$, the function $r \mapsto \sum_{k=1}^{L}{\exp \left (r \Theta^{(k)} \right )}$ is decreasing 
and thus $\{r, \sum_{k=1}^{L}{\exp \left (r \Theta^{(k)} \right )} \leq \gamma_{th} \}=\{r, r> r(\bold{\Theta})\}$ 
where $r(\bold{\Theta})$ solves the non-linear equation $\sum_{k=1}^{L}{\exp \left (r(\bold{\Theta}) \Theta^{(k)} \right )}=\gamma_{th}$. Hence, we get, using the independence of $R$ and $\bold{\Theta}$,
\begin{align}
\ell=\frac{1}{2^N} \mathbb{E}_{\bold{\Theta} \Big{|} \max_i{\Theta_i}<0} \left [ 1-F_R \left ( r(\bold{\Theta})\right )\right ],
\end{align}
where $F_R(\cdot)$ is the CDF of $R$ which is given by
\begin{align}
F_R(r)=\frac{\gamma \left (N/2,r^2/2 \right )}{\Gamma (N/2)}.
\end{align}
Thus, the conditional MC estimator is given as
\begin{align}
\hat {\ell}_{CMC}=\frac{1}{2^N} \left [ 1-F_R \left ( r(\bold{\Theta})\right ) \right ],
\end{align}
where $\bold{\Theta}$ is uniformly distributed on the surface of the N-dimensional unit ball truncated over $\{\max_i{\Theta_i} <0\} $.

The implementation of the conditional MC estimator requires then sampling of $\bold{\Theta}$ truncated over $\{\max_i{\Theta_i} <0\} $. This can be easily performed using the following procedure. First we sample $Y_1,\cdots,Y_N$ independently from the standard Normal distribution, then we set $\Theta_i=- |Y_i|/|| \bold{Y}||_{2}$. It can be easily proven that the output of this procedure provides samples of $\bold{\Theta}$ with the desired distribution. 
\newline Regarding the quantity $r(\bold{\Theta})$ which is the solution of the non linear equation $\sum_{k=1}^{L}{\exp \left (r(\bold{\Theta}) \Theta^{(k)} \right )}=\gamma_{th}$, we approximate it via the use of the bisection method. To do that, we need to construct lower and upper bounds of $r(\bold{\Theta})$. Through a simple computation, we have the following inequality
\begin{align}
\frac{\log \left (\gamma_{th}/L \right )}{\Theta^{(L)}} \leq r(\bold{\Theta}) \leq \frac{\log \left (\gamma_{th}/L \right )}{\Theta^{(1)}}.
\end{align}
\subsubsection{Efficiency}
The following proposition provides an efficiency result of the condition MC estimator.
\begin{prop}
\hspace{2mm} The conditional MC estimator is logarithmic efficient for $1<L \leq N$. That is, for all $\epsilon>0$
\begin{align}
\lim_{\gamma_{th} \rightarrow 0} \frac{\mathrm{var} \left [ \hat{\ell}_{CMC}\right ]}{\ell^{2-\epsilon}}=0.
\end{align}
\end{prop}
\begin{proof}
To facilitate the understanding of the proof, we start with the case where $L=N$. Let us first construct a lower bound of $r(\bold{\Theta})$. To do that, we use the convexity of the exponential function as follows
\begin{align}\label{conv}
\sum_{k=1}^{N}{\exp \left (r \Theta_k \right )} \geq \exp \left ( \sum_{k=1}^{N}{\frac{1}{N} (r \Theta_k+\log(N))}\right ).
\end{align}
Equating the right hand side to $\gamma_{th}$ yields the following lower bound of $r(\bold{\Theta})$
\begin{align}
r(\bold{\Theta}) \geq r_{lower}(\bold{\Theta})=\frac{N \log (\gamma_{th}/N)}{\sum_{k=1}^{N}{\Theta_k}}.
\end{align}
Hence, the second moment of the conditional MC estimator is upper bounded as follows
\begin{align}
\nonumber & \mathbb{E}_{\bold{\Theta} \Big{|} \max_i{\Theta_i}<0} \left [\hat{\ell}^2\right ] \\
& \leq \frac{1}{2^{2N}} \mathbb{E}_{\bold{\Theta} \Big{|} \max_i{\Theta_i}<0} \left [ \left (1-\frac{\gamma(N/2,\frac{N^2 (\log (\gamma_{th}/N))^2}{2(\sum_{k=1}^{N}{\Theta_k})^2})}{\Gamma(N/2)} \right )^2\right ].
\end{align}
Now using the fact that $(\sum_{k=1}^{N}{\Theta_k})^2 \leq N \sum_{k=1}^{N}{\Theta_k^2}=N$, we get that
\begin{align}
\mathbb{E}_{\bold{\Theta} \Big{|} \max_i{\Theta_i}<0} \left [\hat{\ell}^2\right ] \leq \frac{1}{2^{2N}} \left (1-\frac{\gamma(N/2,\frac{N(\log(\gamma_{th}/N))^2}{2})}{\Gamma(N/2)} \right )^2.
\end{align}
Through the use of the following asymptotic behavior, see \cite{abramowitz1964handbook},
\begin{align}
\Gamma(s)-\gamma(s,t) \sim  t^{s-1} \exp(-t) \text{,   as   }t\rightarrow +\infty,
\end{align}
we get for a sufficiently small $\gamma_{th}$ the following upper bound 
\begin{align}
\nonumber & \mathbb{E}_{\bold{\Theta} \Big{|} \max_i{\Theta_i}<0} \left [\hat{\ell}^2\right ] \\
& \leq C_3 \left ( \frac{N(\log(\gamma_{th}/N))^2}{2} \right )^{N-2} \exp \left (-N(\log(\gamma_{th}/N))^2 \right ),
\end{align}
where $C_3$ is a constant independent of $\gamma_{th}$. On the other hand, the probability $\ell$ has the following asymptotic behaviour  \cite{gulisashvili2016}:
\begin{align}
\ell \sim C_4 \left (\log (1/\gamma_{th}) \right )^{-\frac{1+N}{2}} \gamma_{th}^{N \log{N}} \exp \left ( -N \left (\log(\gamma_{th}) \right )^2/2\right ).
\end{align}
Therefore, we have for small enough $\gamma_{th}$
\begin{align}
\frac{\mathbb{E}_{\bold{\Theta} \Big{|} \max_i{\Theta_i}<0} \left [\hat{\ell}^2\right ]}{\ell^2} \leq C_5 \left (\log(1/\gamma_{th}) \right )^{3(N-1)}.
\end{align}
This in particular shows that the conditional MC estimator is logarithmic efficient.

Let us extend the proof to the case where $L<N$. Using the inequality $\frac{N}{L}\sum_{k=1}^{L}{\exp \left (r\Theta^{(k)} \right )} \geq \sum_{k=1}^{N}{\exp \left ( r\Theta_k\right )}$, we construct a lower bound of $r(\bold{\Theta})$ given by equating the right hand side of the previous inequality to $\frac{N}{L} \gamma_{th}$. Then, using the same idea as in (\ref{conv}), we get 
\begin{align}
r(\bold{\Theta}) \geq \frac{N \log (\gamma_{th}/L)}{\sum_{k=1}^{N}{\Theta_k}}.
\end{align}
Moreover, the probability $\ell$ is lower bounded by
\begin{align}
\ell \geq \left ( P \left (X_1 \leq \gamma_{th}/L \right )\right )^N.
\end{align}
Using the asymptotic behavior of the right hand side term given in \cite{DBLP:journals/anor/AsmussenBJR11} 
\begin{align}
P \left ( X_1 \leq \gamma_{th}/L\right ) \sim \frac{1}{\sqrt{2\pi}\log \left (L/\gamma_{th} \right )} \exp \left (-\frac{\left ( \log(\gamma_{th}/L)\right )^2}{2} \right )
\end{align}
and following the same steps as for the case $L=N$, we get
\begin{align}
\frac{\mathbb{E}_{\bold{\Theta} \Big{|} \max_i{\Theta_i}<0} \left [\hat{\ell}^2\right ]}{\ell^2} \leq C \left (\log(1/\gamma_{th}) \right )^{4(N-1)}.
\end{align}
Thus, the logarithmic efficiency holds for $L<N$ as well.
\end{proof}
\begin{rem}
\hspace{2mm} The logarithmic efficiency holds when $X_1,\cdots,X_N$ are i.i.d with parameters $\mu$ and $\sigma$. The proof is a simple modification of the above procedure. 
\end{rem}
\section{Numerical Results}
We provide in this section some selected simulations in order to validate the theoretical results and compare the efficiency of the proposed estimators. We define the relative error, i.e. the coefficient of variation using $M$ replicants, of an estimator $\hat{\ell}$ as
\begin{align}
RE(\hat{\ell})=\frac{\sqrt{\mathrm{var} \left [\hat{\ell} \right ]}}{\ell \sqrt{M}}.
\end{align} 
The simulations are performed for three cases: the Pareto, the Weibull, and the Log-normal distributions. Note that the universal IS estimator described in section III-A is denoted by $\hat{\ell}_{IS,u}$ whereas the IS estimators presented in section III-B and section III-C  are denoted by $\hat{\ell}_{IS}$.
\subsection{Pareto Case}
The system parameters in the Pareto case are as follows. The sequence $X_1,\cdots,X_N$ are i.i.d Pareto RVs with parameter $\alpha=1$. We aim to estimate the CDF of the sum of $L=4$ first order statistics with $N=8$  using the estimators $\hat{\ell}_{IS}$ and $\hat{\ell}_{IS,u}$. Note that the variance of $\hat{\ell}_{IS}$ and $\hat{\ell}_{IS,u}$ can be computed using sample variance or directly through the expression  $\ell_1\ell-\ell^2$ given in (\ref{var_is}). The corresponding results are given in Table \ref{tab1}.

\begin{table}[H]
\begin{center}
\caption{CDF of the sum of order statistics for Pareto Case with $N=8$, $L=4$, $\alpha=1$ and $M=5\times 10^5$.}\label{tab1}
\begin{tabular}{|c|c|c|c|c|}
\hline
\multicolumn{1}{c|}{} &\multicolumn{2}{c|}{IS estimator } & \multicolumn{2}{c|}{Universal IS estimator }\\
\hline
$\gamma_{th}$ & $\hat {\ell}_{IS}$ & $RE(\hat{\ell}_{IS}) \%$ & $\hat {\ell}_{IS,u}$ & $RE(\hat{\ell}_{IS,u}) \%$\\
  \hline
$1.5$  &  $2.21 \times 10^{-4}$   & $ 6.06 \times 10^{-2}$ & $2.19 \times 10^{-4}$ & $1.23$ \\
$1$     & $2.06 \times 10^{-5}$   &  $5.18 \times 10^{-2}$ & $2.11 \times 10^{-5}$ & $1.92$\\
$0.5$   &  $2.13 \times 10^{-7}$  & $3.85 \times 10^{-2}$ & $2.09 \times 10^{-7}$ & $3.82$\\
$0.1$   & $1.29 \times 10^{-12}$  &  $1.79 \times 10^{-2}$ & $1.29 \times 10^{-12}$ & $8.51$\\
\hline
\end{tabular}
\end{center}
\end{table} 
Numerical results show that the quantity $RE(\hat{\ell}_{IS})$ is decreasing as we decrease the threshold value $\gamma_{th}$. Hence, $\hat{\ell}_{IS}$ achieves numerically the asymptotically vanishing relative error property which is stronger than the theoretical result of bounded relative error proven in Proposition 2. Moreover, $\hat{\ell}_{IS}$ is much more efficient than $\hat{\ell}_{IS,u}$, which only achieves the bounded relative error as proved in Proposition 1, and the gain in performance is improving as we decrease the threshold values. Thus, while $\hat{\ell}_{IS,u}$ has the feature of being applicable to a wide range of distributions, its efficiency can be significantly improved for a particular choice of distribution.
\subsection{Weibull Case}
We consider the case where the sequence $X_1,\cdots,X_N$ are i.i.d Weibull RVs with parameter $\eta$ and $\alpha$ and we compare the performance of both IS estimators with the conditional MC one. In order to be able to use the IS estimator $\hat{\ell}_{IS}$ described in section III-C, we restrict our analysis to the case where $0<\alpha<1$. Note that we set $\lambda_k=1/L$, $k=1,\cdots,L$. The system parameters are $L=4$ , $N=8$, $\alpha=0.5$, and $\eta=1$. The corresponding results are given in Table \ref{tab2}

\begin{table}[H]
\centering
\caption{CDF of the sum of order statistics for Weibull Case with $N=8$, $L=4$, $\alpha=0.5$, $\eta=1$ and $M=5\times 10^5$.}\label{tab2}
\begin{tabular}{|c|c|c|c|c|c|c|}
\hline
\multicolumn{1}{c|}{} &\multicolumn{2}{c|}{IS estimator } &\multicolumn{2}{c|}{Universal IS estimator } &\multicolumn{2}{c|}{Conditional MC estimator }\\
\hline
$\gamma_{th}$ & $\hat {\ell}_{IS}$ & $RE(\hat{\ell}_{IS}) \%$& $\hat {\ell}_{IS,u}$ & $RE(\hat{\ell}_{IS,u}) \%$&  $\hat {\ell}_{CMC}$ & $RE(\hat{\ell}_{CMC}) \%$\\
\hline
$1$  &  $0.0029$  &  $9.96 \times 10^{-2}$ & $0.0029$ & $0.4$ &$0.0029$ & $0.12$\\
$0.5$ &  $3.37 \times 10^{-4}$ & $0.1$ & $3.37 \times 10^{-4}$& $0.49$ &$3.37 \times 10^{-4}$ & $0.13$\\
$0.1$  & $1.27 \times 10^{-6}$  & $0.11$ &$ 1.27 \times 10^{-6}$ & $0.66$ &$1.27 \times 10^{-6}$ & $0.15$\\
$0.05$  & $ 9.79 \times 10^{-8}$ & $0.11$ &$9.85 \times 10^{-8}$ & $0.71$ &$9.79 \times 10^{-8}$ & $0.16$\\
$0.01$  & $2.06 \times 10^{-10}$ & $0.11$ & $2.06 \times 10^{-10}$ & $0.8$& $2.07 \times 10^{-10}$ & $0.17$\\
$0.005$ & $1.38 \times 10^{-11}$ & $0.11$ & $1.39 \times 10^{-11}$& $0.81$& $1.38 \times 10^{-11}$ & $0.17$\\
\hline
\end{tabular}
\end{table} 
From the values of the relative error, we deduce  that the three estimators yield very accurate estimates of the unknown probability $\ell$. Moreover, we validate that they have bounded relative error which is in accordance with the theoretical results. Furthermore, the above results show that $\hat{\ell}_{IS}$ outperforms $\hat{\ell}_{IS,u}$ and $\hat{\ell}_{CMC}$. 

Let us now analyze the impact of $\alpha$ on the performance of these three estimators. To this end, we set $\alpha=0.8$ and we repeat the simulation using the same system parameters as above. The results are given in Table \ref{tab3}. We observe from these results that increasing $\alpha$ improves the efficiency of $\hat{\ell}_{IS}$ and $\hat{\ell}_{CMC}$ but has a negative effect on the estimator $\hat{\ell}_{IS,u}$. Moreover, we point out that increasing $\alpha$ results in increasing the efficiency of the $\hat{\ell}_{IS}$ compared to $\hat{\ell}_{CMC}$. This is consistent with the fact that for $\alpha=1$, $\hat{\ell}_{IS}$ has zero variance. 
\begin{table}[H]
\centering
\caption{CDF of the sum of order statistics for Weibull Case with $N=8$, $L=4$, $\alpha=0.8$, $\eta=1$ and $M=5\times 10^5.$}\label{tab3}
\begin{tabular}{|c|c|c|c|c|c|c|}
\hline
\multicolumn{1}{c|}{} &\multicolumn{2}{c|}{IS estimator }  &\multicolumn{2}{c|}{Universal IS estimator }  &\multicolumn{2}{c|}{Conditional MC estimator }\\
\hline
$\gamma_{th}$ & $\hat {\ell}_{IS}$ & $RE(\hat{\ell}_{IS}) \%$& $\hat {\ell}_{IS,u}$ & $RE(\hat{\ell}_{IS,u}) \%$&  $\hat {\ell}_{CMC}$ & $RE(\hat{\ell}_{CMC}) \%$\\
\hline
$1.03$ & $3.38 \times 10^{-4}$ & $5.42 \times 10^{-2}$ & $3.41 \times 10^{-4}$ & $1.28$&$3.37 \times 10^{-4}$  &  $0.1$\\
$0.38$ & $1.32 \times 10^{-6}$ & $5.45 \times 10^{-2}$ & $1.29 \times 10^{-6}$& $2.31$& $1.31 \times 10^{-6}$ & $0.12$\\
$0.09$ & $2.10 \times 10^{-10}$ & $5.56 \times 10^{-2}$& $2.22 \times 10^{10}$ & $3.20$& $2.10 \times 10^{-10}$ & $0.13$\\
$0.058$ & $1.35 \times 10^{-11}$ &  $5.59 \times 10^{-2}$ & $1.33 \times 10^{-11}$& $3.49$ &$1.35 \times 10^{-11}$ & $0.13$\\
\hline
\end{tabular}
\end{table} 
Finally, we investigate the impact of varying $L$. To this end, we provide in Table \ref{tab4} and Table \ref{tab5} the results when $L=2$ and $L=6$ while maintaining $N$ fixed. These tables show that the efficiency of the conditional MC estimator is improved as $L$ increases. However, increasing $L$ affects negatively the performance of $\hat{\ell}_{IS}$ and $\hat{\ell}_{IS,u}$. This in particular suggests to opt for $\hat{\ell}_{CMC}$ when $L$ is close to $N$ and for $\hat{\ell}_{IS}$ when $L$ is close to $1$. 
\begin{table}[H]
\centering
\caption{CDF of the sum of order statistics for Weibull Case with $N=8$, $L=2$, $\alpha=0.5$, $\eta=1$ and $M=5\times 10^5$.}\label{tab4}
\begin{tabular}{|c|c|c|c|c|c|c|}
\hline
\multicolumn{1}{c|}{} &\multicolumn{2}{c|}{IS estimator } &  \multicolumn{2}{c|}{Universal IS estimator } &\multicolumn{2}{c|}{Conditional MC estimator }\\
\hline
$\gamma_{th}$ & $\hat {\ell}_{IS}$ & $RE(\hat{\ell}_{IS}) \%$& $\hat {\ell}_{IS,u}$ & $RE(\hat{\ell}_{IS,u}) \%$& $\hat {\ell}_{CMC}$ & $RE(\hat{\ell}_{CMC}) \%$\\
\hline
$0.355$ & $3.38 \times 10^{-4}$ & $4.37 \times 10^{-2}$ &$3.37 \times 10^{-4}$ & $0.28$ &$3.39 \times 10^{-4}$  &  $0.2$\\
$0.07$ & $1.28 \times 10^{-6}$ & $4.41 \times 10^{-2}$ & $1.28 \times 10^{-6}$& $0.34$& $1.28 \times 10^{-6}$ & $0.25$\\
$0.0069$ & $2.03 \times 10^{-10}$ & $4.42 \times 10^{-2}$ & $2.04 \times 10^{-10}$ & $0.37$&  $2.03 \times 10^{-10}$ & $0.28$\\
$0.0035$ & $1.44 \times 10^{-11}$ &  $4.42 \times 10^{-2}$ &$1.45 \times 10^{-11}$ & $0.38$ &$1.44 \times 10^{-11}$ & $0.28$\\
\hline
\end{tabular}
\end{table} 

\begin{table}[H]
\centering
\caption{CDF of the sum of order statistics for Weibull Case with $N=8$, $L=6$, $\alpha=0.5$, $\eta=1$ and $M=5\times 10^5$.}\label{tab5}
\begin{tabular}{|c|c|c|c|c|c|c|}
\hline
\multicolumn{1}{c|}{} &\multicolumn{2}{c|}{IS estimator } & \multicolumn{2}{c|}{Universal IS estimator }& \multicolumn{2}{c|}{Conditional MC estimator }\\
\hline
$\gamma_{th}$ & $\hat {\ell}_{IS}$ & $RE(\hat{\ell}_{IS}) \%$& $\hat {\ell}_{IS,u}$ & $RE(\hat{\ell}_{IS,u}) \%$& $\hat {\ell}_{CMC}$ & $RE(\hat{\ell}_{CMC}) \%$\\
\hline
$0.55$ & $3.38 \times 10^{-4}$ & $0.17$ & $3.39 \times 10^{-4}$& $0.56$ &$3.39 \times 10^{-4}$  &  $9.88 \times 10^{-2}$\\
$0.11$ & $1.26 \times 10^{-6}$ & $0.18$ &$1.27 \times 10^{-6}$ & $0.79$&$1.26 \times 10^{-6}$ & $0.11$\\
$0.011$ & $2.02 \times 10^{-10}$ & $0.19$ &$2.01 \times 10^{-10}$ &$0.97$ &$2.02 \times 10^{-10}$ & $0.12$\\
$0.0055$ & $1.35 \times 10^{-11}$ &  $0.19$ &$1.36 \times 10^{-11}$ & $1$&$1.34 \times 10^{-11}$ & $0.12$\\
\hline
\end{tabular}
\end{table}

\subsection{Log-normal Case}
We consider the case where the sequence $X_1,\cdots,X_N$ are i.i.d Log-normal RVs with parameter $\mu$ and $\sigma$ and we aim to detect the region for which the conditional MC estimator $\hat{\ell}_{CMC}$ outperforms the universal IS estimator $\hat{\ell}_{IS,u}$ and vice versa. The simulation parameters are $L=4$, $N=8$, $\mu=0$ and $\sigma=2$. The corresponding results are in Table \ref{tab6}.

\begin{table}[H]
\begin{center}
\scriptsize
\caption{CDF of the sum of order statistics for Log-normal Case with $N=8$, $L=4$, $\mu=0$, $\sigma=2$ and $M=10^6$.}\label{tab6}
\begin{tabular}{|c|c|c|c|c|}
\hline
\multicolumn{1}{c|}{}& \multicolumn{2}{c|}{Universal IS estimator } &\multicolumn{2}{c|}{Conditional MC estimator } \\
\hline
$\gamma_{th}$ &  $\hat {\ell}_{IS,u}$ & $RE(\hat{\ell}_{IS,u}) \%$ & $\hat {\ell}_{CMC}$ & $RE(\hat{\ell}_{CMC}) \%$\\
\hline
$1$  & $8.31 \times 10^{-5}$ & $0.68$ & $8.31 \times 10^{-5}$ &  $0.34$\\
$0.5$ & $1.91 \times 10^{-6}$ & $1.27$ & $1.90 \times 10^{-6}$  &  $0.99$\\
$0.3$  & $7.07 \times 10^{-8}$ & $2.11$& $7.00 \times 10^{-8}$ & $2.10$\\
$0.15$  &$3.90 \times 10^{-10}$ &$4.37$ & $3.92 \times 10^{-10}$ & $5.41$\\
\hline
\end{tabular}
\end{center}
\end{table} 
This table reveals that both estimators yield accurate estimates in the considered range of probability values. More precisely, the values of the relative error in Table \ref{tab6} indicate that $10^6$ samples are sufficient to ensure a precise estimate of $\ell$ in the region between $10^{-10}$ and $10^{-5}$. Note also that for values of $\ell$ that are larger than approximately $7 \times 10^{-6}$, the conditional MC estimator has less variance than the universal IS one. However, as we decrease the threshold, the universal IS estimator becomes more efficient than the conditional MC estimator. 

In Table \ref{tab7} and Table \ref{tab8}, we vary $L$ in order to study its impact. The same conclusion as in the Weibull case are deduced. In fact, for fixed $N$, the closer is $L$ to $N$, the better (respectively the worse) is  the performance of the conditional MC estimator (respectively the universal IS estimator). This observation suggests to work with the conditional MC estimator when $L$ is close to $N$ and with the universal IS estimator when $L$ is close to $1$.
\begin{table}[H]
\begin{center}
\scriptsize
\caption{CDF of the sum of order statistics for Log-normal Case with $N=8$, $L=2$, $\mu=0$, $\sigma=2$ and $M=10^6$.}\label{tab7} 
\begin{tabular}{|c|c|c|c|c|}
\hline
\multicolumn{1}{c|}{} &\multicolumn{2}{c|}{Universal IS estimator }&\multicolumn{2}{c|}{Conditional MC estimator } \\
\hline
$\gamma_{th}$ &  $\hat {\ell}_{IS,u}$ & $RE(\hat{\ell}_{IS,u}) \%$ & $\hat {\ell}_{CMC}$ & $RE(\hat{\ell}_{CMC}) \%$\\
\hline
$0.65$    &  $8.26 \times 10^{-5} $ &  $0.31$    & $8.30 \times 10^{-5}$ &  $0.4$\\
$0.315$   & $1.85 \times 10^{-6}$  &   $0.45$  & $1.87 \times 10^{-6}$  &  $1.28$\\
$0.185$   & $ 6.72 \times 10^{-8}$   & $0.6 $&$7.02 \times 10^{-8}$ & $2.97$\\
$0.0908$  &   $3.78 \times 10^{-10}$&  $0.9$ & $3.77 \times 10^{10}$ & $7.43$\\
\hline
\end{tabular}
\end{center}
\end{table} 

\begin{table}[H]
\begin{center}
\scriptsize
\caption{CDF of the sum of order statistics for Log-normal Case with $N=8$, $L=8$, $\mu=0$, $\sigma=2$ and $M=10^6$.}\label{tab8} 
\begin{tabular}{|c|c|c|c|c|}
\hline
\multicolumn{1}{c|}{} &\multicolumn{2}{c|}{Universal IS estimator } & \multicolumn{2}{c|}{Conditional MC estimator } \\
\hline
$\gamma_{th}$ &  $\hat {\ell}_{IS,u}$ & $RE(\hat{\ell}_{IS,u}) \%$ & $\hat {\ell}_{CMC}$ & $RE(\hat{\ell}_{CMC}) \%$\\
\hline
$0.635$  & $1.91 \times 10^{-6}$  & $2.05$ & $1.94 \times 10^{-6}$ &  $0.69$\\
$0.386$ & $6.85 \times 10^{-8}$ &  $3.86$ & $6.98 \times 10^{-8}$  &  $1.29$\\
$0.195$  & $3.39 \times 10^{-10}$ & $9.95$ & $3.57 \times 10^{-10}$ & $2.89$\\
\hline
\end{tabular}
\end{center}
\end{table} 
\subsection{Improvement of the Conditional MC Estimator}
From the smoothness of the conditional MC estimator, it may be interesting to employ the Quasi MC method and investigate whether it leads to further improvement. We consider, as an example, the case of the Weibull distribution whose corresponding conditional MC estimator is given in (\ref{cmc_weibull}) and can be written as $\hat{\ell}_{CMC}=h_1(\bold{S})$. We use \cite[Algorithm 3.19 and 3.22]{opac-b1132466} in order to map a uniform RV $\bold{U}$ over the N-dimensional cube into the random vector $\bold{S}$ by $\bold{S}=h_2(\bold{U})$. With this transformation, the quantity $\ell$ can be expressed as
\begin{align}
\ell=\mathbb{E} \left [ h(\bold{U}) \right ]
\end{align}
with $h=h_1 \circ h_2$. The idea of quasi MC is to consider deterministic quasirandom points and estimate $\ell$ using the sample mean. The objective is to improve the convergence rate to $\mathcal{O} \left (M^{-1/2-\delta} \right )$  instead of $\mathcal{O} \left (M^{-1/2} \right )$ when using i.i.d uniform random points over the N-dimensional cube. However, in order to be able to estimate the error, we consider instead the randomized quasi MC method in which the quasi random points $U_i$, $i=1,\cdots, M$ are now random points, see \cite[Algorithm 2.3]{opac-b1132466} and \cite{owen1997}. Therefore, $\ell$ is estimated by
\begin{align}
\hat{\ell}_{RQMC}=\frac{1}{M}\sum_{i=1}^{M}{h(U_i)}
\end{align}
Since $U_i$ are dependent, we can not estimate the variance of  $\hat{\ell}_{RQMC}$ by sample variance. To overcome such a problem, we produce $m$ independent copies $\hat{\ell}_{RQMC,k}$, $k=1,\cdots,m$, of $\hat{\ell}_{RQMC}$ and estimate $\ell$ by the sample mean of these copies. An estimate of the variance of $\hat{\ell}_{RQMC}$ is then given by the sample variance of these $m$ copies. 

In Fig. \ref{fig_RQMC}, we plot the square root of the variance of $\hat{\ell}_{RQMC}$ as a function of $M$ as well as the MC error rate which is $\mathcal{O} \left (M^{-1/2} \right )$. This figure shows that the randomized quasi MC estimator has a better  rate of convergence than the MC method. In fact, through data fitting, the square root of the variance of $\hat{\ell}_{RQMC}$ decreases with a rate equal approximately to $-1$. Thus, this result ensures a further improvement of the conditional MC estimator in terms of computational effort. 

\begin{figure}[h]
\centering
\setlength\figureheight{0.34\textwidth}
\setlength\figurewidth{0.50\textwidth}
\input{error_RQMC_v2.tikz}
\caption{Convergence Rate of randomized Quasi MC estimator as a function of $M$.}
\label{fig_RQMC}
\end{figure}
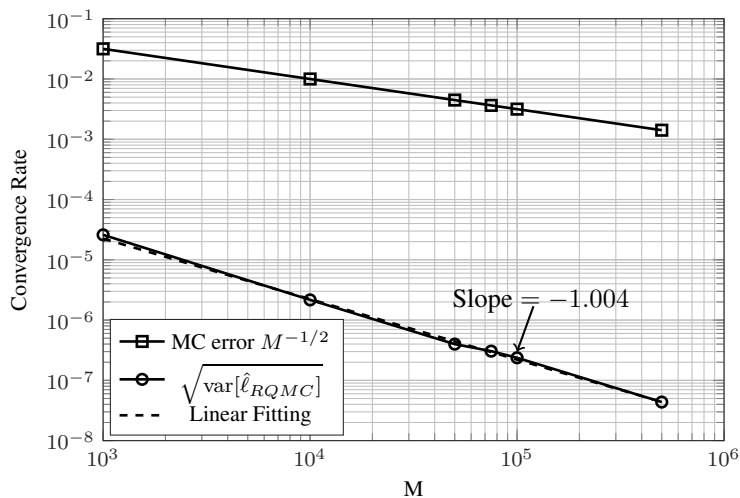

\section{Conclusion}
We developed in this paper two efficient variance reduction techniques in order to estimate the cumulative distribution function of the sum of order statistics. This applies to outage probability computation at the output of receivers with generalized selection combining scheme combined with either maximum ratio combining or equal gain combining diversity techniques. We first provided a universal importance sampling estimator and showed that it achieves the bounded relative error property for most of the well-practical distributions. Moreover, we showed how this approach can be improved if we settle for particular distributions. We also provided a conditional Monte Carlo estimator that has the bounded relative error in the Generalized Gamma case and the logarithmic efficiency in the Log-normal case. Moreover, we studied numerically the efficiency of these estimators and identified the regions in which each estimator performs better than the others. Finally, we showed numerically that the conditional Monte Carlo estimator can be further improved using the randomized quasi Monte Carlo method.  

\bibliography{References}
\bibliographystyle{IEEEtran}
\end{document}

%% file: error_RQMC_v2.tikz
%
%
%
%
\begin{tikzpicture}
\scalefont{0.7}
\begin{loglogaxis}[%
width=\figurewidth,
height=\figureheight,
scale only axis,
every outer x axis line/.append style={darkgray!60!black},
every x tick label/.append style={font=\color{darkgray!60!black}},
xmin=1000, xmax=1000000,
xminorticks=true,
xlabel={M},
xmajorgrids,
xminorgrids,
every outer y axis line/.append style={darkgray!60!black},
every y tick label/.append style={font=\color{darkgray!60!black}},
ymin=1e-08, ymax=0.1,
yminorticks=true,
ylabel={Convergence Rate},
ymajorgrids,
yminorgrids,
grid style={solid},
legend style={at={(0.010856448509367,0.01188527036781)},anchor=south west,draw=darkgray!60!black,fill=white,align=left}]
\addplot [
color=black,
solid,
line width=1.0pt,
mark size=2.0pt,
mark=square,
mark options={solid},
]
coordinates{
 (1000,0.0316227766016838)(10000,0.01)(50000,0.00447213595499958)(75000,0.00365148371670111)(100000,0.00316227766016838)(500000,0.0014142135623731) 
};
\addlegendentry{MC error $M^{-1/2}$};
\addplot [
color=black,
solid,
line width=1.0pt,
mark size=2.0pt,
mark=o,
mark options={solid},
]
coordinates{
 (1000,2.5980904419116e-05)(10000,2.16443574404311e-06)(50000,3.99135377196282e-07)(75000,3.05218691275327e-07)(100000,2.3682855982152e-07)(500000,4.36567770113759e-08) 
};
\addlegendentry{$\sqrt{\mathrm{var} [\hat{\ell}_{RQMC}]}$};
\addplot [
color=black,
dashed,
line width=1.0pt,
]
coordinates{
 (1000,2.24959092672924e-05)(10000,2.22845337371268e-06)(50000,4.42759379716974e-07)(75000,2.94682629943518e-07)(100000,2.20751443287144e-07)(500000,4.38599134513659e-08) 
};
\addlegendentry{Linear Fitting};
\draw[thick,arrows={->}] (axis cs:1.2e05,1.7e-06) -- (axis cs:1.0e05,0.3e-06);
\node[black,below] at (axis cs:1.3e05,4.5e-06){{\small $\text{Slope}=-1.004$}};
\end{loglogaxis}
\end{tikzpicture}%

%% file: Journal_single_column.bbl
\begin{thebibliography}{10}
\providecommand{\url}[1]{#1}
\csname url@samestyle\endcsname
\providecommand{\newblock}{\relax}
\providecommand{\bibinfo}[2]{#2}
\providecommand{\BIBentrySTDinterwordspacing}{\spaceskip=0pt\relax}
\providecommand{\BIBentryALTinterwordstretchfactor}{4}
\providecommand{\BIBentryALTinterwordspacing}{\spaceskip=\fontdimen2\font plus
\BIBentryALTinterwordstretchfactor\fontdimen3\font minus
  \fontdimen4\font\relax}
\providecommand{\BIBforeignlanguage}[2]{{%
\expandafter\ifx\csname l@#1\endcsname\relax
\typeout{** WARNING: IEEEtran.bst: No hyphenation pattern has been}%
\typeout{** loaded for the language `#1'. Using the pattern for}%
\typeout{** the default language instead.}%
\else
\language=\csname l@#1\endcsname
\fi
#2}}
\providecommand{\BIBdecl}{\relax}
\BIBdecl

\bibitem{Yang:2011:OSW:2132708}
H.-C. Yang and M.-S. Alouini, \emph{Order Statistics in Wireless
  Communications: Diversity, Adaptation, and Scheduling in MIMO and OFDM
  Systems}, 1st~ed.\hskip 1em plus 0.5em minus 0.4em\relax New York, NY, USA:
  Cambridge University Press, 2011.

\bibitem{895025}
Y.~Ma and C.~C. Chai, ``Unified error probability analysis for generalized
  selection combining in {N}akagami fading channels,'' \emph{IEEE Journal on
  Selected Areas in Communications}, vol.~18, no.~11, pp. 2198--2210, Nov.
  2000.

\bibitem{7328688}
{N. Ben Rached and A. Kammoun and M.-S. Alouini and R. Tempone}, ``Unified
  importance sampling schemes for efficient simulation of outage capacity over
  generalized fading channels,'' \emph{IEEE Journal of Selected Topics in
  Signal Processing}, vol.~10, no.~2, pp. 376--388, Mar. 2016.

\bibitem{5605378}
S.~S. Nam, M.-S. Alouini, and H.~C. Yang, ``An {MGF}-based unified framework to
  determine the joint statistics of partial sums of ordered random variables,''
  \emph{IEEE Transactions on Information Theory}, vol.~56, no.~11, pp.
  5655--5672, Nov. 2010.

\bibitem{5784177}
S.~S. Nam, M.~O. Hasna, and M.-S. Alouini, ``Joint statistics of partial sums
  of ordered exponential variates and performance of {GSC} {RAKE} receivers
  over {R}ayleigh fading channel,'' \emph{IEEE Transactions on Communications},
  vol.~59, no.~8, pp. 2241--2253, Aug. 2011.

\bibitem{7953495}
S.~S. Nam, Y.~C. Ko, and M.-S. Alouini, ``New closed-form results on ordered
  statistics of partial sums of {G}amma random variables and its application to
  performance evaluation in the presence of {N}akagami fading,'' \emph{IEEE
  Access}, vol.~5, pp. 12\,820--12\,832, 2017.

\bibitem{6596692}
M.-S. Alouini and M.~K. Simon, ``Application of the {D}irichlet transformation
  to the performance evaluation of generalized selection combining over
  {N}akagami-m fading channels,'' \emph{Journal of Communications and
  Networks}, vol.~1, no.~1, pp. 5--13, Mar. 1999.

\bibitem{6292935}
I.~S. Ansari, F.~Yilmaz, M.-S. Alouini, and O.~Kucur, ``On the sum of {G}amma
  random variates with application to the performance of maximal ratio
  combining over nakagami-m fading channels,'' in \emph{in Proc. of the IEEE
  13th International Workshop on Signal Processing Advances in Wireless
  Communications (SPAWC)}, Cesme, Turkey, Jun. 2012, pp. 394--398.

\bibitem{4231253}
M.~D. Yacoub, ``The $\kappa-\mu$ distribution and the $\eta-\mu$
  distribution,'' \emph{IEEE Antennas and Propagation Magazine}, vol.~49,
  no.~1, pp. 68--81, Feb. 2007.

\bibitem{1275712}
N.~Beaulieu and Q.~Xie, ``An optimal {L}ognormal approximation to {L}ognormal
  sum distributions,'' \emph{IEEE Transactions on Vehicular Technology},
  vol.~53, no.~2, pp. 479--489, {M}ar. 2004.

\bibitem{4814351}
M.~Di~Renzo, F.~Graziosi, and F.~Santucci, ``Further results on the
  approximation of {L}og-normal power sum via {P}earson type {IV} distribution:
  a general formula for log-moments computation,'' \emph{IEEE Transactions on
  Communications}, vol.~57, no.~4, pp. 893--898, {A}pr. 2009.

\bibitem{1665128}
J.~Filho and M.~Yacoub, ``Simple precise approximations to {W}eibull sums,''
  \emph{IEEE Communications Letters}, vol.~10, no.~8, pp. 614--616, {A}ug.
  2006.

\bibitem{1388722}
J.~Hu and N.~Beaulieu, ``Accurate simple closed-form approximations to
  {R}ayleigh sum distributions and densities,'' \emph{IEEE Communications
  Letters}, vol.~9, no.~2, pp. 109--111, {F}eb. 2005.

\bibitem{asmussen2016exponential}
S.~Asmussen, J.~L. Jensen, and L.~Rojas-Nandayapa, ``Exponential family
  techniques for the lognormal left tail,'' \emph{Scandinavian Journal of
  Statistics}, vol.~43, no.~3, pp. 774--787, Sep. 2016.

\bibitem{gulisashvili2016}
A.~Gulisashvili and P.~Tankov, ``Tail behavior of sums and differences of
  log-normal random variables,'' \emph{Bernoulli}, vol.~22, no.~1, pp.
  444--493, 2016.

\bibitem{botev_SLN}
Z.~Botev, R.~Salomone, and D.~MacKinlay, ``Accurate computation of the
  distribution of sums of dependent log-normals with applications to the
  black-scholes model,'' \emph{arXiv preprint arXiv:1705.03196}, 2017.

\bibitem{Nadhir_SLN}
M.-S. Alouini, N.~{Ben Rached}, A.~Kammoun, and R.~Tempone, ``On the efficient
  simulation of the left-tail of the sum of correlated {L}og-normal variates,''
  \emph{arXiv preprint arXiv:1705.07635}, 2017.

\bibitem{Stuber:2001:PMC:368633}
G.~L. St\"{u}ber, \emph{Principles of Mobile Communication, 2nd Edition.}\hskip
  1em plus 0.5em minus 0.4em\relax Norwell, MA, USA: Kluwer Academic
  Publishers, 2001.

\bibitem{journals/twc/NavidpourUK07}
S.~M. Navidpour, M.~Uysal, and M.~Kavehrad, ``{BER} performance of free-space
  optical transmission with spatial diversity.'' \emph{IEEE Transactions on
  Wireless Communications}, vol.~6, no.~8, pp. 2813--2819, {A}ug. 2007.

\bibitem{1512431}
N.~Sagias and G.~Karagiannidis, ``Gaussian class multivariate weibull
  distributions: theory and applications in fading channels,'' \emph{IEEE
  Transactions on Information Theory}, vol.~51, no.~10, pp. 3608--3619, {O}ct.
  2005.

\bibitem{opac-b1132466}
D.~P. Kroese, T.~Taimre, and Z.~I. Botev, \emph{Handbook of Monte Carlo
  methods}.\hskip 1em plus 0.5em minus 0.4em\relax N.J: Wiley, 2011.

\bibitem{opac-b1123521}
S.~Asmussen and P.~W. Glynn, \emph{Stochastic simulation : algorithms and
  analysis}, ser. Stochastic modelling and applied probability.\hskip 1em plus
  0.5em minus 0.4em\relax New York: Springer, 2007.

\bibitem{7769235}
N.~{Ben Rached}, A.~Kammoun, M.~S. Alouini, and R.~Tempone, ``A unified
  moment-based approach for the evaluation of the outage probability with noise
  and interference,'' \emph{IEEE Transactions on Wireless Communications},
  vol.~16, no.~2, pp. 1012--1023, Feb 2017.

\bibitem{Botev:2013:SNR:2466677.2466683}
Z.~I. Botev, P.~L'Ecuyer, G.~Rubino, R.~Simard, and B.~Tuffin, ``Static network
  reliability estimation via generalized splitting,'' \emph{INFORMS J. on
  Computing}, vol.~25, no.~1, pp. 56--71, Jan. 2013.

\bibitem{gradshteyn2007}
I.~S. Gradshteyn and I.~M. Ryzhik, \emph{Table of integrals, series, and
  products}, 7th~ed.\hskip 1em plus 0.5em minus 0.4em\relax Elsevier/Academic
  Press, Amsterdam, 2007.

\bibitem{DBLP:conf/wsc/BlanchetJR08}
J.~H. Blanchet, S.~Juneja, and L.~Rojas{-}Nandayapa, ``Efficient tail
  estimation for sums of correlated {L}ognormals,'' in \emph{in Proc. od the
  Winter Simulation Conference, Miami, Florida, USA, Dec.}, 2008, pp. 607--614.

\bibitem{abramowitz1964handbook}
M.~Abramowitz and I.~Stegun, \emph{Handbook of Mathematical Functions: With
  Formulas, Graphs, and Mathematical Tables}, ser. Applied mathematics
  series.\hskip 1em plus 0.5em minus 0.4em\relax Dover Publications, 1964.

\bibitem{DBLP:journals/anor/AsmussenBJR11}
S.~Asmussen, J.~H. Blanchet, S.~Juneja, and L.~Rojas{-}Nandayapa, ``Efficient
  simulation of tail probabilities of sums of correlated {L}ognormals,''
  \emph{Annals {OR}}, vol. 189, no.~1, pp. 5--23, 2011.

\bibitem{owen1997}
A.~B. Owen, ``Scrambled net variance for integrals of smooth functions,''
  \emph{Ann. Statist.}, vol.~25, no.~4, pp. 1541--1562, 1997.

\end{thebibliography}
